\documentclass[11pt]{article}\textwidth 6.5in\textheight 51pc
\usepackage{amssymb}\usepackage[colorlinks]{hyperref}\usepackage{color}
\usepackage{mathrsfs}\usepackage[normalem]{ulem}\usepackage{amsthm}
\usepackage{amsmath}
\usepackage{amssymb}
\usepackage{verbatim}
\usepackage{braket}
\usepackage{setspace}
\usepackage{hyperref}

	\newtheorem{thr}{Theorem} 
\newtheorem{lmm}{Lemma}
\newtheorem{cor}{Corollary}
 
\newtheorem{prp}{Proposition}

\newtheorem{dff}{Definition}

\newtheorem{clm}{Claim}
\newtheorem{exm}{Example}

\newcommand\ceil[1]{{\lceil#1\rceil}}

\newcommand{\lea}{<^{+}}
\newcommand{\gea}{>^{+}}
\newcommand{\eqa}{=^{+}}
\newcommand{\lem}{<^{\ast}}
\newcommand{\gem}{>^{\ast}}
\newcommand{\eqm}{=^{\ast}}

\newcommand\Q{\mathbb{Q}}
\newcommand\R{\mathbb{R}}

\newcommand\N{\mathbb{N}}
\newcommand\BT{\{0,1\}}
\newcommand\FS{\BT^*}
\newcommand\IS{\BT^\infty}
\newcommand\FIS{\BT^{*\infty}}

\newcommand\ii{\mathbf{i}}
\newcommand\om{\overline{\m}}
\newcommand\oI{\overline{\I}}

\newcommand\K{{\mathbf K}} 

\newcommand\I{{\mathbf I}}

\newcommand\m{{\mathbf m}}

\begin{document}

\newcommand{\Tr}{\mathrm{Tr}}
\author {Samuel Epstein\footnote{JP Theory Group. samepst@jptheorygroup.org}}
\title{\vspace*{-3pc} On the Algorithmic Information Between Probabilities} \date{\today}\maketitle

\begin{abstract}
	We extend algorithmic conservation inequalities to probability measures. The amount of self information of a probability measure cannot increase when submitted to randomized processing. This includes (potentially non-computable) measures over finite sequences, infinite sequences, and $T_0$, second countable topologies. One example is the convolution of signals over real numbers with probability kernels. Thus the smoothing of any signal due  We show that given a quantum measurement, for an overwhelming majority of pure states, no meaningful information is produced. 
\end{abstract}
\tableofcontents
\section{Introduction}
We prove conservation of probabilities over successively general spaces. This includeds finite sequences, infinite sequences, and $T_0$, second countable topologies. Conservation of probabilities over the case of finite and infinite sequences follow directly from conservation inequalities over random processing in individual sequences \cite{Levin84,Levin74,Vereshchagin21,Gacs21}. However there is benefit in revisiting these results in the context of manipulations of probabilities. Probabilities are ubiquitious in mathematics, such as the result of quantum measurements, as detailed in this paper. This is particular true when the results are generalized to arbitrary topologies. Information between probability measures is achieved through a mapping from the general topology to infinite sequences and then applying the information function between individual sequences. We use the set of reals as an example and then show conservation of information over computable convolutions. One example is the smoothing of a signal due to a Gaussian function, which results in degradation of self algorithmic information. 

We also show how to lower bound information between probabilities over general spaces with information between probabilities over finite sequences using uniformly enumerable disjoint open sets. We provide an means to upper bound the probabilities between general spaces using computable non-probabilistic measure covers. We look at the average information between measures by using probability measures over spaces of measures.

The advantage to the topological approach used in this paper is that a very general topology can be used. The only assumption needed is that the topology needs to have the $T_0$ property and a computable countable basis. Typical requirements in computability theory such as compactness or metrizability are not needed. In addition this work deals with all measures, not just computable ones. This is analogous to how the mutual information term between infinite sequences is well defined over uncomputable inputs. 

\begin{figure}[h!]
	\begin{center}
\begin{tabular}{|c|l|}
\hline
$\i(x:y)$ & Information between finite sequences\\
\hline
$\ii(p:q)$ & Information between probabilities over finite sequences\\
\hline
$\I(\alpha:\beta)$ & Information between infinite sequences\\
\hline
$\I(P:Q)$ & Information between probabilities over infinite sequences\\
\hline
$\mathcal{I}(\mathcal{P}:\mathcal{Q})$ & Information between probabilities over general spaces\\
\hline
\end{tabular}
\caption{Information terms.}
\end{center}
\end{figure}
\vspace*{-1.0cm}
\section{Probabilities over Sequences}
The function $\K(x|y)$ is the conditional prefix free Kolmogorov complexity. The algorithmic probability is $\m(x|y)$. The mutual information of two finite sequences is $\i(x:y)=\K(x)+\K(y)-\K(x,y)$. $[A]=1$ if the mathematical statement $A$ is true. Otherwise $[A]=0$. Let $\langle x\rangle = 1^{\|x\|}0x$ be a self delimiting encoding of $x$. 

\begin{dff}[Information, Discrete Semi-Measures]$ $\\
	For semi-measures $p$ and $q$ over $\FS$, $\ii(p:q)=\log\sum_{x,y\in\FS}2^{\i(x:y)}p(x)q(y)$.
\end{dff}
The previous definition also applies to semi-measures over $\N$.
\begin{lmm}[\cite{Levin84}]
	\label{lmm:consfunc}
	For partial recursive function $f:\FS\rightarrow\FS$, $\i(f(x):y)\lea \i(x:y)+\K(f)$.
\end{lmm}
\begin{lmm}
	\label{lem:consran}
	Let $\psi_a$ be an enumerable semi-measure, semi-computable relative to $a$.\\ 
	$\sum_c2^{\i(\langle a,c\rangle:b)}\psi_a(c)\lem 2^{\i(a:b)}/\m(\psi)$.
\end{lmm}
\begin{proof}
	This requires a slight modification of the proof of Proposition 2 in \cite{Levin84}, by requiring $\psi$ to have $a$ as auxilliary information. For completeness, we reproduce the proof. We need to show $\m(a,b)/(\m(a)\m(b))\gem \sum_c (\m(a,b,c)/(\m(b)\m(a,c)))\m(\psi)\psi_a(c)$, or $\sum_c (\m(a,b,c)/\m(a,c))\m(c|a)\lem \m(a,b)/\m(a)$, since $\m(c|a)\gem\m(\psi)\psi_a(c)$. Rewrite it $\sum_c\m(c|a)\m(a,b,c)/\m(a,c)\lem \m(a,b)/\m(a)$ or $\sum_c\m(c|a)\m(a)\m(a,b,c)/\m(a,c)\lem \m(a,b)$. The latter is obvious since $\m(c|a)\m(a)\lem \m(a,c)$ and $\sum_c\m(a,b,c)\lem \m(a,b)$.
\end{proof}\newpage
\begin{prp}
\label{prp:discomp}
	For enumerable semi-measures $p$, $q$, $\ii(p:q)\lea\i(\langle p\rangle:\langle  q \rangle)$.
\end{prp}
\begin{proof}
	Let $T$ be a Turing machine, that when given an encoding of a lower semi-computable probabiliy $p$ and an input $x$, lower enumerates $p(x)$. $\ii(p:q)=\log\sum_{x,y}2^{\i(x:y)}T_p(x)T_q(y)$. Using Lemmas \ref{lem:consran} and \ref{lmm:consfunc},	
\begin{align*}	
	&\ii(p:q)\\
	\lea&\log\sum_{x,y}2^{\i(\langle x,p\rangle:y)}T_p(x)T_q(y)\\
	\lea &\log \sum_y2^{\i(\langle p\rangle:y)}q(y)/\m(T)\\
	\lea &\log \sum_y2^{\i(\langle p\rangle:\langle y,q\rangle)}q(y)/\m(T)\\
	\lea& \log 2^{\i(\langle p\rangle:\langle q\rangle)}/\m(T)^2\\
	\lea& \i(\langle p\rangle:\langle q\rangle).
	\end{align*}
\end{proof}

\begin{exm}$ $\\
	\vspace*{-0.5cm}
	\begin{itemize}
	\item  In general, a probability $p$, will have low $\ii(p:p)$ if it has large measure on simple strings, or low measure on a large number of complex strings, or some combination of the two.
	\item If probability $p$ is concentrated on a single string $x$, then $\ii(p:p)=\K(x)$. 
	\item The uniform distribution $U_n$ over strings of length $n$ has self information equal to (up to an additive constant) $\K(n)$. This is because due to Proposition \ref{prp:discomp}, $\ii(U_n:U_n)\lea \K(n)$ and using Lemma \ref{lmm:consfunc}, $\ii(U_n:U_n)=\log\sum_{x,y\in\BT^n}2^{\i(x:y)}2^{-2n}\gea \log \sum_{x,y\in\BT^n}2^{\i(n:n)}2^{-2n}\gea \K(n)$.
\item There are semi-measures that have infinite self information. Let $\alpha_n$ be the $n$ bit prefix of a Martin L\"of random sequence $\alpha$ and $n\in [2,\infty)$.  Semi-measure $p(x)=[x=\alpha_n]n^{-2}$ has 
$\ii(p:p)=\infty$. 
\item The universal semi-measure $\m$ has no self information. 
\item Another example is a probability $p$ where for some $x\in\BT^n$, $p(xy)=2^{-n}$ if $\|y\|=n$, and 0 otherwise. Using Proposition \ref{prp:discomp}, $\ii(p:p)\lea \K(\langle p\rangle)\lea \K(x)$. In addition, using Lemma \ref{lmm:consfunc}, $\ii(p:p)=\log\sum_{xy,xz}2^{\i(xy:xz)}2^{-2n}\gea \log\sum_{xy,xz}2^{\i(x:x)}2^{-2n}\eqa \K(x)$. So the self information of $p$ is equal to $\K(x)$.
\item In general the information between probabilities can be arbitrarily smaller than the information between their encodings. For example, take an arbitrarily large random string $x$, and the probability $p(0)=0.x$, and $p(1)=1-p(0)$. Thus $\ii(p:p)\ll \K(p)$.
\item There exists probabilities $p$ and $q$ such that $\ii(p:p)\ll \ii(p:q)$. Take a large random string $y$ and let $p(0)=0.5$ and $p(y)=0.5$ and $q(y)=1$.
\end{itemize}
\end{exm}\newpage

\begin{dff}[Channel]
	A channel $f:\FS\times\FS\rightarrow\R_{\geq 0}$ has $f(\cdot|x)$ being a probability measure over $\FS$ for each $x\in\FS$. For probability $p$, channel $f$, $fp(x)=\sum_z f(x|z)p(z)$.
\end{dff}
\begin{exm}[Uniform Spread]
An example channel $f$ has $f(\cdot|x)$ be the uniform distribution over strings of length $\|x\|$. This is a cannonical spread function. Thus if $p$ is a probability measure concentrated on a single string, then $\ii(p:p)=\K(x)$, and $\ii(fp:fp)\eqa \K(\|x\|)$. Thus $f$ results in a decrease of self-information of $p$. This decrease of information occurs over all probabilities and computable channels.
\end{exm}

\begin{thr}
	For probabilities $p$ and $q$ over $\FS$, computable channel $f$, $\ii(fp:q)\lea\ii(p:q)$.
\end{thr}
\begin{proof} Using Lemma \ref{lmm:consfunc},
 $\ii(fp:q)$ $=\log\sum_{x,y}2^{\i(x:y)}\sum_zf(x|z)p(z)q(y)$\\ $\lea\log\sum_{y,z
 }q(y)p(z)\sum_x2^{\i((x,z):y)}f(x|z)$. Using Lemma \ref{lem:consran}, $\ii(fp:q)\lea\log\sum_{z,y}q(y)p(z)2^{\i(z:y)}$ $\eqa \ii(p:q)$.	
\end{proof}
\begin{dff}[Information, Infinite Sequences, \cite{Levin74}] $ $\\
For $\alpha,\beta\in\IS$,
	$\I(\alpha:\beta)=\log\sum_{x,y\in\FS}\m(x|\alpha)\m(y|\beta)2^{\i(x:y)}$.
	\end{dff}
\begin{prp}[Folklore]
	\label{prp:folk}
	For $x,y\in\FS$, $\I(\langle x\rangle0^\infty:\langle y\rangle0^\infty)\eqa\i( x: y)$.
\end{prp}
\begin{proof}
	$\I(\langle x\rangle 0^\infty:\langle y\rangle 0^\infty)> \log \m( x|\langle x\rangle 0^\infty)\m( y|\langle y\rangle 0^\infty)2^{\i( x: y)}\eqa \i( x: y)$. For the other direction, using Lemmas \ref{lmm:consfunc} and \ref{lem:consran},
	\begin{align*}
&\I(\langle x\rangle 0^\infty):\langle y\rangle 0^\infty))\\
=&\log\sum_{c,d}\m(c|\langle x\rangle 0^\infty)\m(d|\langle y\rangle 0^\infty)2^{\i(c:d)}\\
\eqa&\log\sum_{c,d}\m(c|x)\m(d|y)2^{\i(c:d)}\\
\lea& \log\sum_{c,d}\m(c| x)\m(d|y)2^{\i(\langle c,x\rangle:d)}\\
\lea &\log\sum_{d}\m(d|y)2^{\i(x:d)}\\
 \lea &\log\sum_{d}\m(d|y)2^{\i(x,\langle d,y\rangle)}\\
 \lea &\log2^{\i(x:y)}.
\end{align*}
\end{proof}
\begin{dff}[Information, Probabilities over Infinite Sequences]
For probabilities $P$, $Q$ over infinite sequences. $\I(P:Q)=\log\int2^{\I(\alpha:\beta)}dP(\alpha)dQ(\beta)$.
\end{dff}
By Carath\'{e}odory's theorem, a measure over $\IS$ can be associated with a function $F:\FS\rightarrow\R_{\geq 0}$, where $F(\emptyset)=1$ and $F(x)=F(x0)+F(x1)$. A probability $P$ is computable if its corresponding function $F$ is computable. The encoding of a computable probability $P$, is equal to $\langle P\rangle=\langle F\rangle$. This term means every possible encoding of $\langle F\rangle$, over all $F$ that computes $P$. Thus if we say $\i(\langle P\rangle:y)>a$, then this means all encoding of $P$ have at least $a$ mutual information with $y$. 
\begin{exm}[Information over Cylinders]
Let $P$ be the measure defined by $F(y)=[x\sqsubseteq y]2^{-\|y\|+\|x\|}$. Thus $P$ is the uniform measure over all sequences that start with $x$. Let $\mathcal{U}$ be the uniform measure over $\IS$. $\I(P:P)=\log \int 2^{\I(\alpha:\beta)}dP(\alpha)dP(\beta)=\int 2^{\I(x\alpha:x\beta)}d\mathcal{U}(\alpha)d\mathcal{U}(\beta)$. For all $x\in\FS$, $\alpha,\beta\in\IS$, $\I(x\alpha:x\beta)> \log \m(x|x\alpha)\m(x|x\beta)2^{\i(x:x)}\gea \K(x)-2\K(\|x\|)$. Thus $\I(P:P)\gea \K(x)-2\K(\|x\|)$. This inequality holds for any probability $P$ whose support is restriced to the cylinder set $x\IS$.
\end{exm}
\begin{thr}[\cite{Vereshchagin21,Levin74}]
	\label{thr:infsimp}
	$\I(A(\alpha):\beta)\lea \I(\alpha:\beta)$, where $A$ is an algorithm and $A(\alpha)$ produces an infinite sequence.
\end{thr}
\begin{thr}[ \cite{Vereshchagin21,Levin74}]
	\label{thr:inf}
	$\int2^{\I\langle (\alpha,\gamma\rangle:\beta)}dP_\gamma(\alpha)< 2^{\I(
		\gamma:\beta)}+c_P$.
\end{thr}
\begin{prp}
	\label{prp:infcomp}
	For computable probabilities $P$ and $Q$ over $\IS$, $\I(P:Q)\lea \i(\langle P\rangle:\langle Q\rangle)$.
\end{prp}
\begin{proof}
	Let $T$ be a program that on input $\langle R\rangle$ for some computable probability $R$, and some string $x$, outputs $R(x)$ to arbitrary precision. Using Theorems \ref{thr:infsimp} and \ref{thr:inf},
	\begin{align*} 
	&2^{\I(P:Q)}\\
	=&\int2^{\I(\alpha:\beta)}dT_{\langle P\rangle}(\alpha)dT_{\langle Q\rangle}(\beta)\\
	\lem &\int2^{\I(\langle \alpha, P\rangle:\beta)}dT_{\langle P\rangle}(\alpha)dT_{\langle Q\rangle}(\beta)\\
	\lem& \int2^{\I(\langle  P\rangle:\beta)}dT_{\langle Q\rangle}(\beta) \\
	\lem& \int2^{\I(\langle  P\rangle:\langle \beta, Q\rangle)}dT_{\langle Q\rangle}(\beta)\\
	\lem& 2^{\I( \langle P\rangle:\langle Q\rangle)}. 
	\end{align*}
	The theoreom follows from Proposition \ref{prp:folk}.
\end{proof}

\begin{dff}
	A random transition is of the form $\Lambda:\IS\times\IS\rightarrow\R_{\geq 0}$ where each $\Lambda(\cdot|\alpha)$ is a semi-measure over $\IS$ for each $\alpha\in\IS$ and for each measurable set $B$ in the Borel algebra of $\IS$, $\Lambda(B|\cdot)$ is a measurable function. For random transition $\Lambda$, probability $P$, $\Lambda P(\alpha) = \int_{\IS} \Lambda(\alpha|\beta)dP(\beta)$. A random transition $\Lambda$ is computable if the semi-measure $\Lambda(\cdot|\alpha)$ is uniformly computable given oracle access to $\alpha$.
\end{dff}

\begin{thr}
	\label{thr:genprobcomp}
	For probabilities $P$, $Q$, computable random transition $\Lambda$, $\I(\Lambda P:Q)\lea \I(P:Q)$. 
\end{thr}
\begin{proof}
	\begin{align*}
	&2^{\I(\Lambda P:Q)}\\
	=&\int_\beta \int_\alpha 2^{\I(\alpha:\beta)} d\Lambda P(\alpha)dQ(\beta)\\
		=&\int_\beta\left(\int_\alpha\left(\int_\gamma  2^{\I(\alpha:\beta)} \Lambda(\alpha|\gamma)dP(\gamma)\right)d\alpha\right)dQ(\beta).
	\end{align*} 
	Using Theorems \ref{thr:infsimp} and \ref{thr:inf},
	\begin{align*}
	 &2^{\I(\Lambda P:Q)}\\\
	 \lem&\int_\beta\left(\int_\alpha\left(\int_\gamma 2^{\I(\langle \alpha,\gamma\rangle:\beta)} \Lambda(\alpha|\gamma)dP(\gamma)\right)d\alpha\right)dQ(\beta)\\
	 \eqm& \int_\beta\left(\int_\gamma\left(\int_\alpha 2^{\I(\langle \alpha,\gamma\rangle:\beta)} d\Lambda(\alpha|\gamma)\right) dP(\gamma)\right)dQ(\beta)\\ \lem&\int_\beta\int_\gamma 2^{\I(\gamma:\beta)} dP(\gamma)dQ(\beta)\\
	 \lem &2^{\I(P:Q)}.
	 \end{align*}
\end{proof}

\section{Probabilities Over General Spaces}
We extend conservation to Borel measures over $T_0$, second countable topologies. We restrict our attention to such topologies which can be represented by a tuple $(X,\mathcal{B},\nu)$ where $X$ is a set, $\nu$ is a countable basis for $X$ where $\nu=(\nu(1),\nu(2),\dots)$, and $\mathcal{B}$ is the Borel algebra formed from the topology. Because of the $T_0$ property, each point $x\in X$ is uniquely defined by the basis sets which contain it. 

\begin{dff}
	\label{dff:infoiso}
We define the following measurable injection $\pi$ from $X$ to $\IS$.  For $\alpha\in X$, let $\pi(\alpha)_i = [\alpha\in \nu(i)]$. For $x\in\FS$, let $\sigma(x)=\pi^{-1}(x\IS)$ be the corresponding measurable set in $X$ associated with $x\in\FS$. 
\end{dff}

Let $\mathcal{R}$ be the smallest ring formed from $\sigma$. By Carath\'{e}odory's theorem, we can associate each pre-measure $\mu$ over $\mathcal{R}$ with a unique Borel measure $\mathcal{P}$ over $\mathcal{B}$, such that its restriction to $\mathcal{R}$ is equal to $\mu$. Thus for each measure $\mathcal{P}$, we can associate a function $F:\FS\rightarrow\R_{\geq 0}$, such that $F(x) =\mathcal{P}\left(\sigma(x)\right)$. The probability measure $\overline{\mathcal{P}}$ over $\IS$ associated with $F$, is called the dual of $\mathcal{P}$. The probability measure $\mathcal{P}$ is computable if $\overline{\mathcal{P}}$ is computable. If $\mathcal{P}$ is computable, then $\langle \mathcal{P}\rangle=\langle \overline{\mathcal{P}}\rangle$.
\begin{clm}
	\label{clm}
For lower semi-continuous $f:\IS\rightarrow \R_{\geq 0}\cup\infty$, probability measure $\mathcal{P}$, \\ $\int_X f(\pi(\alpha))d\mathcal{P}(\alpha)=\int_{\IS}f(\alpha)d\overline{\mathcal{P}}(\alpha)$.
\end{clm}
\begin{proof}
	Let $g:\FS\rightarrow\R_{\geq 0}$, such that $g(x)=\min_{x\sqsubset \alpha}f(\alpha)$ and $f(\alpha)=\sup_{x\sqsubset\alpha}g(x)$. By the definition of integration
	\begin{align*} &\int_{\IS}f(\alpha)d\overline{\mathcal{P}}(\alpha)
	=\lim_{n\rightarrow\infty}\sum_{x\in \BT^n}g(x)\overline{\mathcal{P}}(x)
	= \lim_{n\rightarrow \infty}\sum_{x\in\BT^n} g(x)\mathcal{P}(\sigma(x))
	=\int_X f(\pi(\alpha))d\mathcal{P}(\alpha).
	\end{align*}
\end{proof}
\begin{dff}[Information of Probabilities, General Topology]
	Given two measures $\mathcal{P}$ and $\mathcal{Q}$ over topology $(X,\mathcal{B},\nu)$, their mutual information is $\mathcal{I}(\mathcal{P}:\mathcal{Q})=\log\int 2^{\I(\pi(\alpha):\pi(\beta))}d\mathcal{P}(\alpha)
	d\mathcal{Q}(\beta)$.
\end{dff}\newpage

\begin{prp}
	\label{prp:compgen}
	If $\mathcal{P}$ and $\mathcal{Q}$ are computable, then $\mathcal{I}(\mathcal{P}:\mathcal{Q})\lea \i(\langle\mathcal{P}\rangle:\langle \mathcal{Q}\rangle)$.
\end{prp}

\begin{proof}
	By Claim \ref{clm},
	$$\mathcal{I}(\mathcal{P}:\mathcal{Q})=\int_X\int_X 2^{\I(\pi(\alpha):\pi(\beta))}d\mathcal{P}(\alpha)d\mathcal{Q}(\beta)=\int_{\IS}\int_{\IS}2^{\I(\alpha:\beta)}d\overline{\mathcal{P}}(\alpha)d\overline{\mathcal{Q}}(\beta).$$ The proposition then follows from Proposition \ref{prp:infcomp}.
\end{proof}

A random transition $\Lambda :X\times X\rightarrow \R_{\geq 0}$, is a functiom such that $\Lambda(\cdot|\beta)$ is a semi-measure over $X$ for each $\beta\in X$, and $\Lambda(B|\cdot)$ is a measurable function for each measurable set $B\in\mathcal{B}$. A random transition $\Lambda$ has a dual $\overline{\Lambda}$ random transition in the Cantor space where measurable set $M\subseteq\IS$, $\beta\in \IS$, $\overline{\Lambda}(M|\beta)=\Lambda(\pi^{-1}(M)|\pi^{-1}(\beta))$. If $\pi^{-1}(\beta)$ doesn't exist, then $\overline{\Lambda}(M|\beta)=0$. The set of all such $\beta$ with no $\pi^{-1}(\beta)$ is a measurable set of $\IS$. A random transition is computable if its dual is computable.

\begin{prp}
	\label{prp:transr}
	For lower semi-continuous $f:\IS\rightarrow\R_{\geq 0}\cup\infty$, $\gamma\in X$,\\ $\int_{X} f(\pi(\alpha))d\Lambda(\alpha|\gamma) = \int_{\IS}f(\alpha)d\overline{\Lambda}(\alpha|\pi(\lambda))$.	
\end{prp}
\begin{proof}
	Let $g:\FS\rightarrow\R_{\geq 0}$, such that $g(x)=\min_{x\sqsubset \alpha}f(\alpha)$ and $f(\alpha)=\sup_{x\sqsubset\alpha}g(x)$. \\
	\begin{align*} &\int_{X} f(\pi(\alpha))d\Lambda(\alpha|\gamma)\\
	=& \lim_{n\rightarrow \infty}\sum_{x\in\BT^n} g(x)\Lambda(\sigma(x)|\gamma)\\
	=&\lim_{n\rightarrow\infty}\sum_{x\in \BT^n}g(x)\overline{\Lambda}(x|\pi(\gamma))\\
	=&\int_{\IS} f(\alpha)d\overline{\Lambda}(\alpha|\pi(\gamma)).
	\end{align*}
\end{proof}
\begin{thr}
	For probabilities $\mathcal{P}$, $\mathcal{Q}$, computable random transition $\Lambda$, $\mathcal{I}(\Lambda \mathcal{P}:\mathcal{Q})\lea \mathcal{I}(\mathcal{P}:\mathcal{Q})$. 
\end{thr}
\begin{proof}
	\begin{align*}
	2^{\mathcal{I}(\Lambda \mathcal{P}:\mathcal{Q})}=&\int_X \int_X 2^{\I(\pi(\alpha):\pi(\beta))} d\Lambda \mathcal{P}(\alpha)d\mathcal{Q}(\beta)\\
	=&\int_X\left(\int_X\left(\int_X  2^{\I(\pi(\alpha):\pi(\beta))} \Lambda(\alpha|\gamma)d\mathcal{P}(\gamma)\right)d\alpha\right)d\mathcal{Q}(\beta)\\
	=&\int_X\left(\int_X\left(\int_X  2^{\I(\pi(\alpha):\pi(\beta))} d\Lambda(\alpha|\gamma)\right)d\mathcal{P}(\gamma)\right)d\mathcal{Q}(\beta)
	\end{align*}
	Using Proposition \ref{prp:transr},
	\begin{align*}
	2^{\mathcal{I}(\Lambda P:Q)}=&\int_X\left(\int_X\left(\int_{\IS} 2^{\I(\alpha:\pi(\beta))} d\overline{\Lambda}(\alpha|\pi(\gamma))\right)d\mathcal{P}(\gamma)\right)d\mathcal{Q}(\beta).
	\end{align*}
	Using Theorem \ref{thr:infsimp},
	\begin{align*}
2^{\mathcal{I}(\Lambda \mathcal{P}:\mathcal{Q})}\lem&\int_X\left(\int_X\left(\int_{\IS} 2^{(\langle\alpha,\pi(\gamma)\rangle:\pi(\beta))} d\overline{\Lambda}(\alpha|\pi(\gamma))\right)d\mathcal{P}(\gamma)\right)d\mathcal{Q}(\beta) \beta.
\end{align*}
Using Theorem \ref{thr:inf},
\begin{align*}
	2^{\mathcal{I}(\Lambda \mathcal{P}:\mathcal{Q})} \lem&\int_X\int_X 2^{\I(\pi(\gamma):\pi(\beta))} d\mathcal{P}(\gamma)d\mathcal{Q}(\beta)\eqm 2^{\mathcal{I}(\mathcal{P}:\mathcal{Q})}.
	 \end{align*}
\end{proof}

\begin{clm}
	For probability measures $p$, $q$ over finite sequences, infinite sequences, or general spaces, if $p$ is computable, $\I(p:q)\lea \K(p)$.
\end{clm}
\begin{proof}
We prove it for the finite sequence case, and the other cases follow similarly. By Lemma \ref{lem:consran}, $\ii(p:q)=\log \sum_{x,y}2^{\i(x:y)}p(x)q(y)\lea \log \sum_y2^{\i(\langle p\rangle :y)}q(y) \lea \log 2^{\K(p)}$.
\end{proof}
\begin{exm}[Convolutions on the real line]
	\label{exm:rl}
One example topology to be used throughout this paper is the real line, $(\R,\mathcal{B}_\R,\nu)$, where $\nu$ is the set of all intervals $\{(a,b)\}$, $a,b\in\Q$. For such topology, the Gaussian distribution $\mathcal{N}(\mu,\sigma^2)$ is computable, and thus by Theorem \ref{thr:genprobcomp}, has self information bounded by $\mathcal{I}(\mathcal{N}(\mu,\sigma^2):\mathcal{N}(\mu,\sigma^2))\lea \K(\mu,\sigma^2)$. Similarly the self-information of parameterized distributions over $\R$ will be less than the complexity of their encoded parameters. 

A convolution of a probability $\mathcal{P}$ over $\R$ with a probability kernel $\mathcal{F}$ produces a new probability $(\mathcal{P}\star\mathcal{F})(x) = \int_{-\infty}^\infty \mathcal{P}(y)\mathcal{F}(x-y)dy$. Convolution is a random transition, and it is computable if $\mathcal{F}:\R\rightarrow\R_{\geq 0}$ is computable. The pdf of the sum of two random variables is the convolution of their respective pdfs.
\end{exm}
\begin{cor}
For probability $\mathcal{P}$ over $\R$, computable probability kernel $\mathcal{F}:\R\rightarrow\R_{\geq 0}$,\\ $\mathcal{I}(\mathcal{P}\star\mathcal{F}\,{:}\mathcal{P}\star\mathcal{F})\,\lea \mathcal{I}(\mathcal{P}\,{:}\,\mathcal{P})$.
\end{cor}
	
Let $\mathcal{G}\sim\mathcal{N}(0,\sigma^2)$ be a Gaussian distribution over $\R$. Thus convolution of a signal (probability measure) $\mathcal{P}$ with $\mathcal{G}$ results in smoothing of $\mathcal{P}$, proportional to $\sigma^2$. By the above corollary, a smoothing of any signal (computable or not) will result in a decrease of self information.
\section{Information Transfer}
In this section, we show how information between probabilities on general spaces can be related to the information between probabilities of finite sequences through the use of disjoint, enumerable open sets.

\begin{dff}[Information, Infinite Sequences, Extended]
For $x,y\in\FS$, $\om(x|y)=\sum\{2^{-\|p\|}:U_y(p)=x\textrm{ and at most $\|y\|$ bits are read from the auxilliary tape}\}$. For $x\in\FS$, $\alpha\in\IS$, $\om(x|\alpha)=\m(x|\alpha)$. For $\alpha,\beta\in\FS\cup\IS$, $\oI(\alpha:\beta)=\log\sum_{x,y\in\FS}\om(x|\alpha)\om(y|\beta)2^{\i(x:y)}$. If $x\sqsubseteq y$, $\oI(x:\beta)\leq \oI(y:\beta)$ and  $\lim_{n\rightarrow\infty}\oI(\alpha[0..n]:\beta)=\oI(\alpha:\beta)=\I(\alpha:\beta)$.
\end{dff}

\begin{dff}[Prefix monotone function]
	 \mbox{$\sqsubseteq$-$\sup$} is the supremum under the partial order of $\sqsubseteq$ on $\FIS$. Say function $\nu\,{:}\,\FS\,{\rightarrow}\,\FS$ has the property that for all $p,q\in\FS$,  $\nu(p)\,{\sqsubseteq}\,\nu(pq)$. Then $\overline{\nu}\,{:}\,\FIS\,{\rightarrow}\,\FIS$ denotes the unique extension of $\nu$, where $\overline{\nu}(p)= \sqsubseteq$-$\sup\, \{\nu(p_{\leq n})\,{:}\,n\,{\leq}\,\|p\|,n\,{\in}\,\N\}$ for all $p\,{\in}\,\FIS$. These extensision are called prefix monotone functions.
\end{dff}
\begin{clm}
	\label{clm:pm}
	For computable prefix-monotone function $A$, $\alpha,\beta\in \FS\cup\IS$, $\oI(A(\alpha):\beta)\lea \oI(\alpha:\beta)$. 
	\end{clm}
\begin{dff}
A series of open sets $\{\phi_i\}$ is enumerable if $\{(i,j) : \sigma(j)\subseteq \phi_i\}$ is enumerable. 
\end{dff}
\begin{thr}
	\label{thr:infotrans}
	For enumeration $\{\phi_i\}$ of disjoint open elements, if two probability measures $\mathcal{P}$ and $\mathcal{Q}$ have corresponding semi-measures $p$ and $q$ over $\N$, where $p(i)\leq\mathcal{P}(\phi_i)$ and $q(i)\leq\mathcal{Q}(\phi_j)$, then $\mathcal{I}(\mathcal{P}:\mathcal{Q})\gea \ii(p:q)$.
\end{thr}
\begin{proof}
Let $A:\FS\cup\IS\rightarrow\FS\cup\IS$, be a computable prefix monotone function mapping each $\alpha\in \IS$ to $\langle i\rangle$ if there exists $j\in\N$ where $\alpha[j]=1$ and $\nu_j\subseteq \phi_i$. All other sequences are $A$-mapped to $\emptyset$. This function is well defined because $\{\phi_i\}$ are disjoint. The function $A$ is computable because each open set $\phi_i\subseteq X$ is uniformly enumerable. By Claims \ref{clm} and \ref{clm:pm}, 
	\begin{align*}
	&	 2^{\mathcal{I}(\mathcal{P}:\mathcal{Q})}\\
	=&\int_{\IS}\int_{\IS}2^{\oI(\alpha:\beta)}d\overline{\mathcal{P}}(\alpha)\overline{\mathcal{Q}}(\beta)\\
\gem&\int_{\IS}\int_{\IS}2^{\oI(A(\alpha):A(\beta))}d\overline{\mathcal{P}}(\alpha)\overline{\mathcal{Q}}(\beta)\\ \\
\gem&\sum_{i,j\in\N}\int_{\alpha:\pi^{-1}(\alpha)\in\phi_i}\int_{\beta:\pi^{-1}(\beta)\in\phi_j}2^{\oI(\langle i\rangle:\langle j\rangle)}d\mathcal{P}(\pi^{-1}\alpha)d\mathcal{Q}(\pi^{-1}\beta)\\
\eqm& \sum_{i,j\in\N} 2^{\oI(\langle i\rangle:\langle j\rangle)}\mathcal{P}(\phi_i)\mathcal{Q}(\phi_j)\\
	\gem &\sum_{i,j}2^{\I(i:j)}p(i)q(j)\\
	\eqm& 2^{\I(p:q)}.
	\end{align*}
\end{proof}
The above theorem is true for the case of infinite sequences. Let $\{\phi_i\}$, $\{\theta_i\}$ be enumerations of disjoint open sets on $\IS$, and $P$, $Q$, be probability measures on $\IS$. Then for semi measures $p(i)\leq P(\phi_i)$ and $q(i)\leq Q(\theta_i)$, we have $\I(P:Q)\gea \ii(p:q)$.
\begin{exm}
For any probability measure $\mathcal{P}$ over $\R$, $\mathcal{I}(\mathcal{P}:\mathcal{P})\gea \K(n)+2\log \mathcal{P}(\mathcal{U}(n,n+1))$, where $\mathcal{U}(n,n+1)$ is the uniform distribution. This is due to the enumeration of $\{(n,n+1)\}$.
\end{exm}
\begin{exm}[Random Pulse]
	Let $\alpha\in\IS$ be a Martin L\"{o}f random sequence, where $\K(\alpha[0..n])\gea n$.
	Using the real line topology of Example \ref{exm:rl}, let $\mathcal{N}_n = \mathcal{N}(0.\alpha, n^{-2})$ be a Gaussian distribution centered at a random point between 0 and 1 and with variance inversely proportional to $n$. We show that $\lim_{n\rightarrow\infty}\mathcal{I}(\mathcal{N}_n:\mathcal{N}_n)=\infty$. This limit is known as the (offseted) Dirac delta function.
	
	Let $\nu \in (0,0.5)$ and $\{\phi^m_i\}$ be $2^m$ open intervals  where $\phi_i^m =((i-1)2^{-m},i2^{-m})$, for $i\in[1,2^m]$. Each $i$ can be associated with a string $x$ of length $m$, where $s(x)=i$. For each $m$, there is an $n$, such that $\mathcal{N}_n(\phi_{s(\alpha[m])}^m)>\nu$. Thus, using Theorem \ref{thr:infotrans} (noting that the additive constant is on the order of the complexity of the enumeration) for each $m$, there is an $n$, where $\mathcal{I}(\mathcal{N}_n,\mathcal{N}_n) \gea \K(s(\alpha[m]))+2\log \nu-O(\K(\{\phi_i^m\}))\gea m + 2\log \nu-O(\K(m))$. Thus $\lim_{n\rightarrow\infty}\mathcal{I}(\mathcal{N}_n:\mathcal{N}_n)=\infty$.
\end{exm}
\begin{exm}
	For rational $\sigma^2$, we show that for $n\in\N$, $\mathcal{I}(\mathcal{N}(n,\sigma^2):\mathcal{N}(n,\sigma^2))\eqa \K(n)$. Let $\mathcal{N}_n=\mathcal{N}(n,\sigma^2)$.
	Indeed, due to Proposition \ref{prp:compgen}, $\mathcal{I}(\mathcal{N}_n:\mathcal{N}_n)\lea \K(n,\sigma^2)\eqa \K(n)$. For the other direction let $\{\phi_i\}$ be an enumeration of open intervals, where $\phi_i=(i,i+1)$. Then $\mathcal{N}_n(\phi_n)>c_{\sigma}$, for some $c_\sigma\in \R_{\geq 0}$, dependent solely on $\sigma^2$. Due to Theorem \ref{thr:infotrans} on $\mathcal{N}_n$ there is a semi-measure $p$ over $\N$, where $p(n)=c_\sigma$, and $\mathcal{I}(\mathcal{N}_n:\mathcal{N}_n)\gea \ii(p:p)\gea \K(n)+2\log p(n)\eqa \K(n)$. 
\end{exm}

\section{Computable Covers}
Given a topology $(X,\sigma,\mathcal{B},\nu)$, a (not necessarily probability) measure $\rho$ covers measure $\mu$ if $\rho(B)\geq \mu(B)$ for all measurable sets $B\in\mathcal{B}$.
\begin{thr}
	If computable measures $\mathcal{M}$ and $\mathcal{R}$ cover probability measures $\mathcal{P}$, $\mathcal{Q}$, then $\mathcal{I}(\mathcal{P}:\mathcal{Q})\lea \i(\langle \mathcal{M}\rangle:\langle \mathcal{R}\rangle)+\log \mathcal{M}(X)\mathcal{R}(X)$.
\end{thr}
\begin{proof}
	We define computable probability measures, $m=\mathcal{M}/\mathcal{M}(X)$ and $r=\mathcal{R}/\mathcal{R}(X)$.Then 
	\begin{align*}
	&\mathcal{I}(\mathcal{P}:\mathcal{Q})\\
	\leq& \log\int_X\int_X2^{\I(\pi(\alpha):\pi(\beta))}d\mathcal{M}(\alpha)d\mathcal{R}(\beta)\\
	\lea& \log\mathcal{M}(X)\mathcal{R}(X)+\log\int_X\int_X2^{\I(\pi(\alpha):\pi(\beta))}dm(\alpha)dr(\beta)\\
	\eqm& \log\mathcal{M}(X)\mathcal{R}(X)+\mathcal{I}(m:r).
	\end{align*}
	 Let $F_m$ and $F_r$ be the sets of programs that compute $m$ and $r$, respectively. Using Proposition \ref{prp:compgen}, for all $f_m\in F_m$, $f_r\in F_r$.
\begin{align*}
\mathcal{I}(\mathcal{P}:\mathcal{Q})&\lem \log\mathcal{M}(X)\mathcal{R}(X)+\i(\langle f_m\rangle:\langle f_r\rangle).
\end{align*}
	Let $F_M$ and $F_R$ be programs that compute $\mathcal{M}$ and $\mathcal{R}$ and minimize $\i(\langle F_M\rangle:\langle F_R\rangle)$. Let $f'_m\in F_m$ and $f'_r\in F_R$ be programs that compute $m$ and $r$ by first computing $F_M(x)$ then dividing by $F_M(\emptyset)$, and similarly for $F_R$. Thus it must be that $\K(\langle f'_m\rangle|\langle F_M \rangle)=O(1)$ and similarly, $\K(\langle f'_r\rangle|\langle F_R\rangle)=O(1)$. Using Lemma \ref{lmm:consfunc},
	\begin{align*}
\mathcal{I}(\mathcal{P}:\mathcal{Q})\lea& \log\mathcal{M}(X)\mathcal{R}(X)+\i(\langle m\rangle:\langle r\rangle)\\
\lea& \log\mathcal{M}(X)\mathcal{R}(X)+\i(\langle f'_m\rangle:\langle f'_r\rangle)\\
\lea& \log\mathcal{M}(X)\mathcal{R}(X)+\i(\langle F_M\rangle:\langle F_R\rangle)\\
\lea& \log\mathcal{M}(X)\mathcal{R}(X)+\i(\langle \mathcal{M}\rangle:\langle \mathcal{R}\rangle).
	\end{align*}
\end{proof}
\begin{cor}
	For semi measures $p$ and $q$ over $\FS$, if computable measures $w$ and $r$ over $\FS$ have $p\leq w$ and $q\leq r$, then $\ii(p:q)\lea \log w(\FS)r(\FS)+\i(\langle w\rangle:\langle r\rangle)+\K(\ceil{w(\FS)})+\K(\ceil{r(\FS)})$. 
\end{cor}

\begin{exm}$ $\\
	\vspace*{-0.5cm}
	\begin{itemize}
		\item If probability $\mathcal{P}$ over $\R$ has support limited to $(a,b)$ with $\hat{c}\geq \sup_{a\in \R}\mathcal{P}(a)$, then for $\hat{a}\leq a<b\leq \hat{b}$, $\mathcal{I}(\mathcal{P}:\mathcal{P})\lea 2\log \hat{c}(\hat{b}-\hat{a})+\K(\hat{a},\hat{b},\hat{c})$.
		\item If probability $\mathcal{P}$ over $\R$ is less than the weighted combination $\sum_{i=1}^nc_i\mathcal{Q}_i$ of parametric distributions $Q_i$ (such as a Gaussian or a exponential distribution), then $\mathcal{I}(\mathcal{P}:\mathcal{P})\lea 2\log \sum_i^nc_i+\K(\{c_i,\mathcal{Q}_i\}_{i=1}^n)$.
		\item Let $\mathcal{U}(a,b)$ be the uniform measure over $(a,b)$. If an infinite sequence $\alpha\in\IS$ is encoded into a semi-measure of the form $\mathcal{P}(x) = \sum_{n=2}^\infty \alpha_nn^{-2}\mathcal{U}(n,n+1)(x)$ then it is covered by  $\mathcal{M}(x)=\sum_{n=2}^\infty n^{-2}\mathcal{U}(n,n+1)(x)$ which has negligible self-information. However if $\mathcal{P}$ is normalized to a probability measure, it can have arbitrarily high self-information. An example is $\alpha\in\IS$, $\alpha[i]= [i=n]$, where $n\in\N$ is a large random number.
			\end{itemize}
\end{exm}

\section{Averaged Information}
The average information between probability measures is small, less than the complexity of the averaging. This is true in the discrete and continuous case. For the discrete case, an enumerable sequence of uniformly computable probability measures over a general space is a sequence of measures $\{\mu_i\}$ such that $\overline{\mu}_i(x\IS)$ is uniformly computable with respect to $i$, for all $x\in\FS$.
\begin{thr}
	\label{thr:disavg}
	Let $\mathcal{E}=\{\mu_i\}$ be an enumerable sequence of uniformly computable probability measures over a general space. Let $p$ be a computable measure over $\N$. Then $\mathbf{E}_{i,j\sim p}[2^{\mathcal{I}(\mu_i:\mu_j)}]=O(1)$. 
\end{thr}
\begin{proof}
	The measure $\mu=\sum_ip_i\mu_i$ is computable because for each $\epsilon$, $\overline{\mu}(x)$ can be computed to within $\epsilon$. Thus by Proposition \ref{prp:compgen}, $\mathcal{I}(\mu:\mu)\lea \K(\mu)\lea\K(p,\mathcal{E})=O(1)$. This implies $\mathbf{E}_{i,j\sim P}[2^{\mathcal{I}(\mu_i:\mu_j)}]=O(1)$, because
	\begin{align*}
	2^{\mathcal{I}(\mu:\mu)}
	{=}\int_X\int_X2^{\I(\pi(\alpha):\pi(\beta))}d\mu(\alpha)d\mu(\beta)
	{=}\sum_{i,j}p_ip_j\int_X\int_X2^{\I(\pi(\alpha):\pi(\beta))}d\mu_i(\alpha)d\mu_j(\beta)
	{=}\sum_{i,j}p_ip_j2^{\mathcal{I}(\mu_i:\mu_j)}.
	\end{align*}
\end{proof}
\begin{exm}
	Let $\mathcal{E}$ consist of $2^n$ Gaussians $\mathcal{N}(u,1)$ for $u\in\{1,..,2^n\}$. Let $p$ be the uniform distribution over the first $2^n$ natural numbers. Then by Theorem \ref{thr:disavg}, $\mathbf{E}_{i,j\sim p}[2^{\mathcal{I}(\mathcal{N}(i,1):\mathcal{N}(j,1))}]=O(1)$.
\end{exm}
For the continuous case, we use a random transition between two different measure spaces. This differs from other approaches such as \cite{HoyrupRo09,Gacs21} which defines a metric space of measures. We recall that an random transition from one topology $(X_M,\sigma_M,\mathcal{B}_M,\nu_M)$ to another $(X,\sigma,\mathcal{B},\nu)$ is a measurable function $\Gamma:X_M\times X\rightarrow [0,1]$, such that $\Gamma(\cdot|x_M)$ is a probability measure for each $x_M\in X_M$ and $\Gamma(B|\cdot)$ is a measurable function for all measurable sets $B\in\mathcal{B}$. Thus topology $X_M$ can be seen as a space of probability measures, with each point representing a probability measure. A probability measure $\mathfrak{M}$ over $X_M$ produces an averaged (over $X_M$) probability over $X$, with $\mathcal{P}(\sigma(x))=\int_{X_M}\Gamma(\sigma(x)|\alpha)d\mathfrak{M}(\alpha)$. 
\begin{thr}
	\label{thr:contavg}
	For computable $\mathcal{P}$, $\mathbf{E}_{\alpha,\beta\sim\mathfrak{M}}[2^{\mathcal{I}(\Gamma_\alpha:\Gamma_\beta)}]=O(1)$.
\end{thr}
\begin{proof}
	For computable $\mathcal{P}$, by Proposition \ref{prp:compgen}, $\mathcal{I}(\mathcal{P}:\mathcal{P})\lea \K(\mathcal{P})=O(1)$. This implies\\ $\mathbf{E}_{\alpha,\beta\sim\mathfrak{M}}[2^{\mathcal{I}(\Gamma_\alpha:\Gamma_\beta)}]=O(1)$ because
	\begin{align*}
	&2^{\mathcal{I}(\mathcal{P}:\mathcal{P})}\\
	=&\int_X\int_X2^{\I(\pi(\alpha):\pi(\beta))}d\mathcal{P}(\alpha)d\mathcal{P}(\beta)\\
	=&\int_X\int_X2^{\I(\pi(\alpha):\pi(\beta))}\left(\int_{X_M}\Gamma(\alpha|\xi)d\mathfrak{M}(\xi)\right)d\alpha \left(\int_{X_M}\Gamma(\beta|\zeta)d\mathfrak{M}(\zeta)\right)d\beta\\
	=&\int_{X_M}\int_{X_M}\left(\int_X\int_X2^{\I(\pi(\alpha):\pi(\beta))}\Gamma(\alpha|\xi)\Gamma(\beta|\zeta)d\alpha d\beta\right)d\mathfrak{M}(\xi)d\mathfrak{M}(\zeta)\\
	=&\int_{X_M}\int_{X_M}2^{\mathcal{I}(\Gamma_\xi:\Gamma_\zeta)}d\mathfrak{M}(\xi)d\mathfrak{M}(\zeta)\\
	=&\mathbf{E}_{\alpha,\beta\sim\mathfrak{M}}\left[2^{\mathcal{I}(\Gamma_\alpha:\Gamma_\beta)}\right].\\
	\end{align*}
\end{proof}
\begin{exm}
	We let $X$ be the real line and the random transition $\Gamma(\cdot|u)$ be a Gaussian $\mathcal{N}(u,1)$ with mean $u$. The space of measures is $X_M=[0,1]$, representing all Gaussians with 1 variance with means between 0 and 1. Then for $\mathfrak{M}\sim\mathcal{U}[0,1]$ being the uniform measure between $0$ and $1$, we have $\mathbf{E}_{a,b\sim\mathcal{U}[0,1]}\left[2^{\mathcal{I}(\mathcal{N}(a,1):\mathcal{N}(b,1))}\right]=O(1)$.
\end{exm}

\section{Quantum Mechanics}
Quantum information theory studies the limits of communicating through quantum channels. This section shows the limitations of the algorithmic content of [ure states and their measurements. Given a measurement apparatus $E$, there is only a tiny fraction of quantum pure states on which $E$’s application produces coherent information. This is independent of the number of measurement outcomes of $E$.

In quantum mechanics, given a quantum state $\ket{\psi}$, a measurement, or POVM, $E$ produces a probability measure $E\ket{\psi}$ over strings. This probability represents the classical information produced from the measurement. 
More formally, a POVM $E$ is a finite set of positive definite matrices $\{E_k\}$ such that $\sum_k E_k = I$. For a given density matrix $\sigma$, a POVM $E$ induces a probability measure over strings, where $E\sigma(k) = \Tr E_k\sigma$. This can be seen as the probability of seeing measurement $k$ given quantum state $\sigma$ and measurement $E$.
 
Given a measurement $E$, for an overwhelming majority of pure quantum states $\ket{\psi}$, the probability produced will have no meaningful information, i.e. $\ii(E\ket{\psi}:E\ket{\psi})$ is negligible. 

\begin{thr}
Let $\Lambda$ be the uniform distribution on the unit sphere of an $n$ qubit space. For the universal Turing machine relativized to an encoding of POVM $E$, $\int 2^{\ii(E\ket{\psi}:E\ket{\psi})}d\Lambda = O(1)$.
\end{thr}
The proof for this theorem can be found in \cite{EpsteinQuantum21}. Its form is rather bizarre, in that it uses \textit{upper semi-computable} tests, which is most likely the only place in the algorithmic information theory literature where this occurs.

Another interesting property of quantum mechanics is that the vast majority of quantum pure states themselves will have negligible algorithmic self information $\I_\mathcal{Q}(\ket{\psi}:\ket{\psi})$. For this definition we use the information term introduced in \cite{Epstein19}. From this reference, we get the following result.
\begin{thr}
Let $\Lambda$ be the uniform distribution on the unit sphere of an $n$ qubit space.\\ $\int 2^{\I_\mathcal{Q}(\ket{\psi}:\ket{\psi})}d\Lambda = O(1)$.
\end{thr}
It is an open question as to what spaces of measures will have a majority of its members have negligible self information.

\end{document}